\documentclass[conference]{IEEEtran}
\usepackage{dsfont,graphicx,color,amsfonts,epsf,epsfig,verbatim,amssymb,amsmath,array,cite,amsthm,subfigure,multicol,multirow}
\ifCLASSINFOpdf
\else
\fi

\begin{document}
\title{On the Capacity of Abelian Group Codes Over Discrete Memoryless Channels}

\author{\IEEEauthorblockN{Aria G. Sahebi and S. Sandeep Pradhan}
\IEEEauthorblockA{Department of Electrical Engineering and Computer Science,\\
University of Michigan, Ann Arbor, MI 48109, USA.\\
Email: \tt\small ariaghs@umich.edu, pradhanv@umich.edu}}

\newtheorem{theorem}{Theorem}[section]
\newtheorem{deff}{Definition}[section]
\newtheorem{example}{Example}[section]
\newtheorem{lemma}[theorem]{Lemma}
\newtheorem{prop}[theorem]{Proposition}
\newtheorem{cor}[theorem]{Corollary}

\maketitle

\begin{abstract}
For most discrete memoryless channels, there does not exist a linear code for the channel which uses all of the channel's input symbols. Therefore, linearity of the code for such channels is a very restrictive condition and there should be a loosening of the algebraic structure of the code to a degree that the code can admit any channel input alphabet. For any channel input alphabet size, there always exists an \emph{Abelian group} structure defined on the alphabet. We investigate the capacity of Abelian group codes over discrete memoryless channels and provide lower and upper bounds on the capacity.
\end{abstract}



%
\IEEEpeerreviewmaketitle

\section{Introduction}

Approaching information theoretic performance limits of communication
systems using structured codes has been an area of great interest in recent
years \cite{ahlswede_alg_codes,sandeep_discus,forney_dynamics,fagnani_abelian,dinesh_dsc,dobrushin_group}.
The earlier attempts to design fast encoding and decoding algorithms
resulted in injection of algebraic structures to the coding scheme so that the
channel input alphabets are replaced with algebraic fields and encoders are replaced with
matrices. It is well-known that binary linear codes achieve the capacity of
binary symmetric channels \cite{elias}. More generally, it has also been
shown that $q$-ary linear codes can achieve the capacity of symmetric
channels \cite{dobrushin_group} and linear codes can be used to compress
a source losslessly down to its entropy \cite{korner_marton}. Optimality of
linear codes for certain communication problems motivates the study of
structured codes in general.

In 1979, Korner-Marton showed that for
multiterminal communication problems, the asymptotic  average performance of
linear code ensembles can be superior to that of the standard
code ensembles traditionally used in information theory. In the recent
past, such gains have been shown for a wide class of problems
\cite{phiosof_zamir, dinesh_dsc, nazer_gastpar}. Hence
information-theoretic characterizations of performance of such structured code
ensembles for various communication problems have become important.

The algebraic structure of the code,  however, imposes
certain restrictions on the
encoder.   Linear codes are highly structured and for certain communication
problems such codes cannot be optimal. Moreover, they can only be
constructed on alphabets of certain size (prime power).
Group codes are a class of algebraic-structured codes that are more general
because we can construct such codes over any alphabet, and they
have been shown to outperform unstructured codes in
certain communication settings \cite{dinesh_dsc}. Group codes were first
studied by Slepian \cite{slepian_group} for the Gaussian channel. In
\cite{ahlswede_group}, the capacity of group codes for certain classes of
channels has been computed. Further results on the capacity of group codes
were established in \cite{ahlswede_alg_codes,ahlswede_alg_codes2}. The
capacity of group codes over a class of channels exhibiting symmetries with
respect to the action of a finite Abelian group has been investigated in
\cite{fagnani_abelian}.\\

In this work, we focus on the point-to-point channel coding problem over
general discrete memoryless channels.  The channel input
alphabet is equipped with the structure of an Abelian group. We
characterize the performance of asymptotically good Abelian group
codes over general discrete memoryless channels.
 In particular, we derive
lower and upper bounds on the capacity of Abelian group codes for
communication over such channels.  We use a combination of algebraic and
information-theoretic tools for this task.

The paper is organized as follows. In section \ref{notation}, we introduce
our notation and develop the required background. Section \ref{bounds}
presents the lower and upper bound on the capacity of Abelian group codes.
In section \ref{special_cases} we present two special cases, namely, linear
codes over arbitrary channels and arbitrary Abelian group codes over
symmetric channels where the two bounds match.

\section{Definitions and notation} \label{notation}
\subsubsection{Group Codes}
Given a group $G$, a group code $\mathds{C}$ over $G$ with block length $n$ is any subgroup of $G^n$ \cite{forney_dynamics,algebra_bloch}. A shifted group code over $G$, $\mathds{C}+v$ is a group code $\mathds{C}$ shifted by a fixed vector $v\in G^n$.

\subsubsection{Source and Channel Models}
We consider discrete memoryless and stationary channels used without feedback. We associate two finite sets $\mathcal{X}$ and $\mathcal{Y}$ with the channel. These channels can be characterized by a conditional probability law $W(y|x)$. The set $\mathcal{X}$ admits the structure of a finite abelian group $G$ of the same size. The channel is specified by $(G,\mathcal{Y},W)$. Assuming a perfect source coding block applied prior to the channel coding, the source of information generates messages over the set $\{1,2,\ldots,M\}$ uniformly.

\subsubsection{Achievablility and Capacity}
A transmission system with parameters $(n,M,\tau)$ for reliable communication over a given channel $(G,\mathcal{Y},W)$ consists of an encoding mapping and a decoding mapping $e:\{1,2,\ldots,M\}\rightarrow G^n, \mbox{   }f:G^n\rightarrow\{1,2,\ldots,M\}$
%
such that for all $m=1,2,\ldots,M$,
\begin{align*}
\frac{1}{M}\sum_{m=1}^{M}W^n\left(f(Y^n)\ne
m|X^n=e(m)\right)\le \tau
\end{align*}
Given a channel $(G,\mathcal{Y},W)$, the rate $R$ is said to be achievable if for all $\epsilon>0$ and for all sufficiently large $n$, there exists a transmission system for reliable communication with parameters $(n,M,\tau)$ such that $\frac{1}{n}\log M \ge R-\epsilon,\mbox{   }\tau\le \epsilon$.\\
%
If there is no constraint on the encoder, the maximum achievable rate is called the (Shannon) capacity of the channel and is denoted by $C_{|G|}$ which is known to be equal to $\max_{p_X}I(X;Y)$. $|G|$ denotes the cardinality (size) of the set $G$. We use this notation since only the size and not the structure of the channel input alphabet determines the quantity $C_{|G|}$. In this paper, the encoder is constrained to be \emph{affine} and therefore the code is a shifted group code. We denote the maximum achievable rate of such codes by $C_G$. If the distribution of $X$ is confined to be uniform over $G$, we define $C_{|G|}^U=I(X;Y)$. The capacity of shifted group codes over $H$ which is itself a subgroup of a larger group $G$ is denoted by $C_{H,G}$.
\subsubsection{Typicality}
Consider two random variables $X$
and $Y$ with joint probability density function $p_{X,Y}(x,y)$
over $\mathcal{X}\times\mathcal{Y}$. Let $n$ be an integer and
$\epsilon$ a positive real number. The sequence pair $(x^n,y^n)$
belonging to $\mathcal{X}^n\times \mathcal{Y}^n$ is said to be
jointly $\epsilon$-typical with respect to $p_{X,Y}(x,y)$ if
\begin{align*}
\nonumber \forall a\in\mathcal{X},\mbox{   }\forall b\in\mathcal{Y}:
\left|\frac{1}{n}N\left(a,b|x^n,y^n\right)-p_{X,Y}(a,b)\right|\le
\frac{\epsilon}{|\mathcal{X}||\mathcal{Y}|}
\end{align*}
and none of the pairs $(a,b)$ with $p_{X,Y}(a,b)=0$ occurs in
$(x^n,y^n)$. Here, $N(a,b|x^n,y^n)$ counts the number of
occurrences of the pair $(a,b)$ in the sequence pair $(x^n,y^n)$.
We denote the set of all jointly $\epsilon$-typical sequences
pairs in $\mathcal{X}^n\times \mathcal{Y}^n$ by
$A_\epsilon^n(X,Y)$.\\
Given a sequence $x^n\in \mathcal{X}^n$, the set of conditionally
$\epsilon$-typical sequences $A_\epsilon^n(Y|x^n)$ is defined as
\begin{eqnarray}
A_\epsilon^n(Y|x^n)=\left\{y^n\in \mathcal{Y}^n\left| (x^n,y^n)\in
A_\epsilon^n(X,Y)\right.\right\}
\end{eqnarray}
In our notation, $O(\epsilon)$ is any function of $\epsilon$ such that $\lim_{\epsilon\rightarrow 0}O(\epsilon)=0$.

\section{Bounds on the Capacity of Abelian Group Codes}\label{bounds}
It is a standard fact (see \cite{group_hall} and \cite{algebra_bloch} for example) that any Abelian group $G$ can be decomposed into $\mathds{Z}_{p^r}$ groups in the form $G\cong \bigoplus_{i=1}^{I} \mathds{Z}_{{p_i}^{r_i}}$
%
for some integers $r_i$ and primes $p_i$ for $i=1,2,\cdots,I$ with the possibility of repetitions. Define $R_i=\mathds{Z}_{{p_i}^{r_i}}$ to get $G\cong \bigoplus_{i=1}^{I}R_i$. This means that any element $g$ in the group $G$ can be represented by an $I$-tuple $(g_1,g_2,\cdots,g_I)$ where $g_i\in R_i=\{0,1,\cdots,p_i^{r_i}-1\}$ and this representation preserves the group structure of $G$. Any subgroup $H$ of $G$ can be represented by $H\cong \bigoplus_{i=1}^I p_i^{\theta_i} R_i$.
\subsection{Lower bound}
\begin{theorem}
A lower bound on the Capacity of group codes over the group $G\cong \bigoplus_{i=1}^{I}R_i$ for a discrete memoryless channel $(G,\mathcal{Y},W)$ is given by:
\begin{align*}
C_G\ge \max_{\substack{w_1,\cdots,w_I\\w_1+\cdots+w_I=1}}\min_{H\le G} \sum_{S \mbox{ coset of } H} \frac{|H|}{|G|}\frac{C_{|S|}^U}{w_H}
\end{align*}
where $w_H=\sum_{i=1}^I \frac{r_i-\theta_i}{r_i}w_i$ for $H\cong \bigoplus_{i=1}^I p_i^{\theta_i} R_i$ and $C_{|S|}^U$ is the mutual information between the channel input and output when the input distribution is uniform over the subset $S$ of $G$.
\end{theorem}
The subgroup $S$ of $H$ that achieves the maximum value for $C_{|S|}^U$, is called the optimal subchannel corresponding to the subgroup $H$ and is denoted by $H^*$.\\
\textbf{Proof:} We construct an ensemble of homomorphic encoders over $G$ with block length $n$ and put a uniform distribution over the ensemble. Then we calculate the expected average probability of error over the ensemble and observe that for rates less than $C_G$, the average probability of error can be made arbitrarily small by increasing the block length.
\subsubsection{Construction of the ensemble of codes}
Let $w_i$, $i=1,2,\cdots,I$ be a set of nonnegative rational weights assigned to each module $R_i$ such that $\sum_{i=1}^{I}w_i=1$ and let $k$ be a nonnegative integer so that $w_ik$ is integer for all $i$. For each set of weights, we define an ensemble of codes by taking into account all homomorphisms $\varphi:\bigoplus_{i=1}^I R_i^{w_ik}\rightarrow G^n$. It is known that the image of a homomorphism is a subgroup of the target group \cite{group_hall}; Therefore any such homomorphism defines a group code $\mathds{C}$ over $G$. We add a random dither $v$ to the code to construct a random shifted group code.

Let $m=1,2,\cdots,M$ be the set of messages. Let $k$ be large enough so that a unique message representative $u(m)$ from the set $\bigoplus_{i=1}^I R_i^{w_ik}$ can be assigned to each message $m$. The encoding rule is given by $e(m)=\varphi(u(m))+v$ where $\varphi$ is an arbitrary homomorphism from $\bigoplus_{i=1}^I R_i^{w_ik}$ to $G^n$ and $v$ is a random vector in $G^n$.\\

At the decoder, after receiving the channel output $y$, decode it to the message $m$ if $m$ is the unique message such that $u(m)$ and $y$ are jointly $\epsilon$-typical. Otherwise declare error.\\

The standard generator of the ring $R_i=\mathds{Z}_{p_i^{r_i}}$ is the multiplicative identity of $R_i$. Define $e_{iK}$ to be the generator for the $K$th $R_i$ in $\bigoplus_{i=1}^I R_i^{w_ik}$ for $i=1,\cdots,I$ and $K=1,\cdots w_ik$. Then any element $a\in \bigoplus_{i=1}^I R_i^{w_ik}$ can be represented uniquely as $a=\sum_{i,K} a_{iK}e_{iK}$ where $a_{iK}\in R_i$. This decomposition will help us characterizing homomorphisms from $\bigoplus_{i=1}^I R_i^{w_ik}$ to $G^n$.
\begin{lemma}
Any homomorphism $\varphi:\bigoplus_{i=1}^I R_i^{w_ik}\rightarrow G^n$ can be represented as $\varphi=(\varphi_1,\varphi_2,\cdots,\varphi_n)$ where each $\varphi_N$, $N=1,\cdots,n$ is given by:
\begin{align*}
\phi_N(a)=\sum_{i,K}a_{iK}g_{iK}^N
\end{align*}
for some $g_{iK}^N$'s, $i=1,\cdots,I$, $K=1,\cdots,w_i k$ in $G$.
\end{lemma}
\begin{proof}
Follows from standard algebraic arguments.
\end{proof}
The lemma above facilitates the construction of the ensemble of codes as follows: Take random elements $g_{iK}^N$ from the group $G$ for $n=1,\cdots,N$, $i=1,\cdots,I$ and $k=1,\cdots,w_iK$ and construct the homomorphism $\varphi$ as mentioned in the lemma. Also take a random vector $v$ from $G^n$ and use the encoding rule $e(m)=\varphi(u(m))+v$.\\
The rate of the codes in this ensemble is given by:
\begin{align*}
R=\frac{k}{n}\sum_{i=1}^I w_i\log|R_i|=\frac{k}{n}\sum_{i=1}^I w_i r_i \log p_i
\end{align*}
\subsubsection{Error Analysis}
The expected value of the average probability of word error is given by:
\begin{align*}
&\nonumber\mathds{E}\left\{P_{avg}(err)\right\}=\sum_{m=1}^{M}\frac{1}{M}\sum_{x\in
G^n}P\left(e(m)=x\right) \sum_{\substack{\tilde{m}=1\\\tilde{m}\ne m}}^{M}\\
&\sum_{y\in A_\epsilon^n(Y|x)} \sum_{\tilde{x}\in A_\epsilon^n(X|y)}P\left(e(\tilde{m})=\tilde{x}, Y^n=y|e(m)=x\right)+O(\epsilon)\\
\end{align*}
We need two lemmas to proceed.

\begin{lemma}\label{problemma}
For arbitrary messages $m$ and $\tilde{m}$ and arbitrary vectors $x,\tilde{x}\in G^n$, define $a=u(m)-u(\tilde{m})$ and $h=x-\tilde{x}$. Define $\theta(m,\tilde{m})=(\theta_1,\theta_2,\cdots,\theta_I)$ where $\theta_i$ is the smallest number in $\{0,1,\cdots,r_{i}-1\}$ such that there exists an index $K\in \{1,2,\cdots,w_ik\}$ with the property $a_{iK}\in p_i^{\theta_i}R_i\backslash p^{\theta_i+1}R_i$. Then,
\begin{align*}
&\nonumber P(e(\tilde{m})=\tilde{x}|e(m)=x)=\\
&\left\{ \begin{array}{ll}
\prod_{i=1}^I\frac{1}{p_i^{n(r_i-\theta_i)}}
&\mbox{if $\tilde{x}\in x+\left[\bigoplus_{i=1}^I p_i^{\theta_i} R_i\right]^n$};\\
0 & \mbox{otherwise}.
\end{array}\right.
\end{align*}
Moreover, for a fixed $m$, let $T_\theta(m)$ be the set of all $\tilde{m}$ with $\theta(m,\tilde{m})=(\theta_1,\theta_2,\cdots,\theta_I)$, then
\begin{align*}
|T_\theta(m)|\le\prod_{i=1}^I\left[(p_i^{r_i-\theta_i})^{w_ik}\right]
\end{align*}
\end{lemma}
\begin{proof}
Provided in the appendix.
\end{proof}

\begin{lemma}\label{cosettypicallemma}
Let $y\in \mathcal{Y}^n$ be an arbitrary channel output sequence. For any $x\in A_\epsilon^n(X|y)$, we have
\begin{align*}
&\left|\left(x+\left[\bigoplus_{i=1}^I p_i^{\theta_i} R_i\right]^n\right)\cap A_\epsilon^n(y)\right|\\
&\le \prod_{i=1}^I2^{n\left[H(X_i|Y)-H([X_i]_{\theta_i})+O(\epsilon)\right]}
\end{align*}
where $X_i$ is the $i$th component of the channel input random variable $X$. i.e. $X\longleftrightarrow(X_1,X_2,\cdots,X_I)$ where $X_i\in R_i$ and the random variable $[X_i]_{\theta_i}$ takes values from the set of cosets of $p_i^{\theta_i}R_i$ in $R_i$.
\end{lemma}

\begin{proof}
Provided in the appendix.
\end{proof}

The following lemma presents an upper bound on the average probability of error.
\begin{lemma}\label{upperbound}
The average probability of error over the ensemble is bounded above by:
\begin{align*}
\mathds{E}\left\{P_{avg}(err)\right\}\le \sum_\theta \exp_2\left\{-n\sum_{i=1}^I \left[(r_i-\theta_i)\log p_i\right.\right.&\\
\left.\left.-\frac{w_i k}{n}(r_i-\theta_i)\log p_i-H(X_i|Y)+H([X_i]_{\theta_i}|Y)\right]\right\}&\\
\end{align*}
\end{lemma}

\begin{proof}
Provided in the appendix.
\end{proof}
Each random variable $X_i$ can be represented by a tuple $([X_i]_{\theta_i},[\hat{X}_i]_{\theta_i})$ where $[X_i]_{\theta_i}$ indicates the coset selection and $[\hat{X}_i]_{\theta_i}$ the value selection in the subgroup $p_i^{\theta_i}R_i$ of $R_i$. Note that $[X_i]_{\theta_i}$ and $[\hat{X}_i]_{\theta_i}$ are independent. We get,
\begin{align*}
\mathds{E}\left\{P_{avg}(err)\right\}\le \sum_\theta \exp_2\left\{-n\sum_{i=1}^I \left[(r_i-\theta_i)\log p_i\right.\right.&\\
\left.\left.-\frac{w_i k}{n}(r_i-\theta_i)\log p_i-H([\hat{X}_i]_{\theta_i}|[X_i]_{\theta_i},Y)\right]\right\}&\\
\end{align*}
Therefore, the probability of error can be made arbitrarily small if for all $\theta$,
\begin{align*}
\sum_{i=1}^I \frac{w_i k}{n}(r_i-\theta_i)&\log p_i\\
&\le \sum_{i=1}^I (r_i-\theta_i)\log p_i-H([\hat{X}_i]_{\theta_i}|[X_i]_{\theta_i},Y)
\end{align*}
Let $X$ be a uniform random variable over $G$ and let $H$ be the subgroup of $G$ isomorphic to $\bigoplus_{i=1}^I p_i^{\theta_i} R_i$. The variable $X$ can be thought of as a uniform variable over a random coset of $H$ in $G$. Random selection of the coset is due to the random dither and we prove in Lemma \ref{uniform_dist} that the uniformity of the distribution over the coset is due to the group structure of the code. The variable $X$ can be represented by two random variables $[\hat{X}]_H$ and $[X]_H$ where $[\hat{X}]_H$ is uniform over $H$ and $[X]_H$ has a uniform distribution over cosets of $H$ in $G$ and represents the coset selection. The variable $[\hat{X}]_H$ itself can be represented by a tuple $([\hat{X}_1]_{\theta_1},\cdots,[\hat{X}_I]_{\theta_I})$ where the for each $i$ the random variable $[\hat{X}_i]_{\theta_i}$ is a uniform variable over $p_i^{\theta_i} R_i$ and $[\hat{X_i}]_{\theta_i}$'s are independent from each other and from $[X]_H$. The random variable $[X]_H$ can also be represented by a tuple $([{X}_1]_{\theta_1},\cdots,[{X}_I]_{\theta_I})$ where for each $i$ the random variable $[{X_i}]_{\theta_i}$ is a uniform variable over cosets of $p_i^{\theta_i} R_i$ in $R_i$ and $[{X}_i]_{\theta_i}$'s are independent from each other and from $[\hat{X}]_H$.
\begin{align*}
I([\hat{X}]_H;Y|&[X]_H)\\
&=I\left([\hat{X}_1]_{\theta_1},\cdots,[\hat{X}_{\theta_I}];Y|[{X}_1]_{\theta_1},\cdots,[{X}_{\theta_I}]\right)\\
&=\sum_{i=1}^I I([\hat{X}_i]_{\theta_i},Y|[{X}_i]_{\theta_i})\\
&=\sum_{i=1}^I (r_i-\theta_i)\log p_i-H([\hat{X}_i]_{\theta_i}|[X_i]_{\theta_i},Y)
\end{align*}
Therefore, the achievability condition is equivalent to
\begin{align*}
\sum_{i=1}^I \frac{w_i k}{n}(r_i-\theta_i)\log p_i\le I([\hat{X}]_H;Y|[X]_H)
\end{align*}
Where $H=\bigoplus_{i=1}^I p_i^{\theta_i} R_i$.\\
The rate of the code is given by $R=\frac{k}{n}\sum_{i=1}^I w_i r_i \log p_i$. Therefore, this condition is equivalent to
\begin{align*}
R\cdot\frac{\sum_{i=1}^I w_i(r_i-\theta_i)\log p_i}{\sum_{i=1}^I w_i r_i\log p_i}\le I([\hat{X}]_H;Y|[X]_H)
\end{align*}
Define $w_H=\frac{\sum_{i=1}^I w_i(r_i-\theta_i)\log p_i}{\sum_{i=1}^I w_i r_i\log p_i}$ to get
\begin{align*}
R\le \frac{1}{w_H}I([\hat{X}]_H;Y|[X]_H)
\end{align*}
Note that
\begin{align*}
I([\hat{X}]_H;&Y|[X]_H)\\
&=\sum_{S \mbox{ coset of } H}p([X]_H=S)I([\hat{X}]_H;Y|[X]_H=S)\\
&=\sum_S \frac{|H|}{|G|}C_S^U
\end{align*}
Since this condition must be satisfied for every subgroup $H$ of $G$ and the weights $w_i$ are arbitrary, we conclude that the rate
\begin{align*}
R^*=\max_{\substack{w_1,\cdots,w_I\\w_1+\cdots+w_I=1}}\min_{H\le G} \sum_S \frac{|H|}{|G|}\frac{C_S^U}{w_H}
\end{align*}
is achievable using group codes over $G$. The weights $w_i$ can be represented as $w_i=\frac{\sum_{i=1}^I k w_i(r_i-\theta_i)\log p_i}{\log M}=\sum_{i=1}^I\frac{r_i-\theta_i}{r_i}w'_i$ where $w'_i=\frac{kw_i\log |R_i|}{\log M}$. Since $\sum_{i=1}^I w'_i=1$, the given achievable rate region is equivalent to:
\begin{align*}
R^*=\max_{\substack{w_1,\cdots,w_I\\w_1+\cdots+w_I=1}}\min_{H\le G} \sum_{S \mbox{ coset of } H} \frac{|H|}{|G|}\frac{C_S^U}{w_H}
\end{align*}
where $w_H=\sum_{i=1}^I \frac{r_i-\theta_i}{r_i}w_i$. Here we have replaced $w'_i$'s with $w_i$'s for simplicity of notation.

\subsection{Upper bound}
\begin{deff}
A subgroup $H$ of $G$ is called maximal for the channel $(G,\mathcal{Y},W)$ if for all subgroups $S$ of $H$, $C_{|H^*|}^U\ge C_{|S^*|}^U$.
\end{deff}
\begin{theorem}
An upper bound on the capacity of group codes over the group $G\cong \bigoplus_{i=1}^{I}R_i$ for a memoryless channel $(G,\mathcal{Y},W)$ is given by:
\begin{align*}
C_G\le \max_{\substack{w_1,\cdots,w_I\\w_1+\cdots+w_I=1}}\min_{H \mbox{ maximal}} \max_{S \mbox{ coset of } H} \frac{C_{|S|}^U}{w_H}
\end{align*}
where $w_H=\sum_{i=1}^I \frac{r_i-\theta_i}{r_i}w_i$ for $H\cong \bigoplus_{i=1}^I p_i^{\theta_i} R_i$ and $C_{|S|}^U$ is the mutual information between the channel input and output when the input distribution is uniform over the subset $S$ of $G$.
\end{theorem}
\textbf{Proof:}
\subsubsection{Converse channel coding theorem}
Shannon's inverse channel coding theorem asserts that for rates $R>I(X;Y)$ lossless communication is not possible. For $i=1,2,\cdots,n$, let $X_i$ be the random variable representing the $N$th component of the codewords and $Y_i$ be the corresponding channel output. The rate is bounded above by $R<\frac{1}{n}\sum_{i=1}^n I(X_i,Y_i)$.\\
This theorem admits the generalization to the case where the single letter distribution of $X$ is constrained by the structure of the code. For the case of shifted group codes, the single letter distribution of $X$ can only be uniform on cosets of different subgroups of the underlying group.
\subsubsection{Uniform single letter distribution over cosets}
In the case of linear codes, the single letter distribution over the channel input symbols is confined to be uniform. This holds for group codes also; However, for group codes, it can be uniform over any subgroup of the channel input alphabet.
\begin{lemma}\label{uniform_dist}
For any group code $\mathds{C}\le G^n$ where $G$ is an arbitrary group, uniform multiletter distribution over messages induces a uniform single letter distribution over subgroups of $G$. i.e. the components of the channel input sequence are uniformly distributed over some subgroup of $G$ that varies for different components.
\end{lemma}
\begin{proof}
Without loss of generality we prove that the $n$th component of the codewords form a subgroup $H$ of $G$ and the uniform distribution over codewords induces a uniform distribution over $H$. Let $\{c_1,c_2,\cdots,c_M\}$ be the set of codewords and let $P_{[n,n]}(\mathds{C})=\{c_{1n},c_{2n},\cdots,c_{Mn}\}$ be the set of the $n$th components of the codewords. It has been shown in \cite{forney_dynamics} that $P_{[n,n]}(\mathds{C})$ is a subgroup of $G$. Set $H=P_{[n,n]}(\mathds{C})$ to conclude the first part of the claim.\\
Next, we need to show that the single letter distribution over $H$ is uniform. Let $H=\{h_1=0,h_2,\cdots,h_{|H|}\}$; then the lemma claims that the number of occurrences of each $h_i$ in the sequence $c_{1n},c_{2n},\cdots,c_{Mn}$ is the same. Let $\mathds{C}_{[1,n-1]}$ be the set of all codewords that are zero at the $n$th component. It is known that $\mathds{C}_{[1,n-1]}$ forms a normal subgroup of $\mathds{C}$ and $\mathds{C}/\mathds{C}_{[0,n-1]}\cong P_{[n,n]}(\mathds{C})=H$ \cite{forney_dynamics}. Therefore, $\frac{|\mathds{C}|}{|\mathds{C}_{[1,n-1]}|}=|H|$. The number of occurrences of $h_1=0$ in the sequence $c_{1n},c_{2n},\cdots,c_{Mn}$ is equal to $|\mathds{C}_{[1,n-1]}|$. For each $h^*\in H$, there exists a codeword $c^*\in \mathds{C}$ ending with $h^*$, and since $\mathds{C}$ is a group code, it is closed under addition and therefore $c^*+\mathds{C}_{[1,n-1]}$ is a subset of $\mathds{C}$. Since the codewords are distinct, the set $c^*+\mathds{C}_{[1,n-1]}$ contains $|\mathds{C}_{[1,n-1]}|$ codewords ending with $h^*$. We conclude that for each $h^*\in H$ the existence of at least $|\mathds{C}_{[1,n-1]}|$ codewords ending with $h^*$ is guaranteed. The equality $\frac{|\mathds{C}|}{|\mathds{C}_{[1,n-1]}|}=|H|$ imposes the number of occurrences of each $h^*$ to be equal to $|\mathds{C}_{[1,n-1]}|$. i.e. The single letter distribution over $H$ is uniform in the $n$th position.
\end{proof}
\begin{lemma}
For any shifted group code $\mathds{C}+v$ over $G$, uniform multiletter distribution over messages induces a uniform single letter distribution over cosets of subgroups of $G$.
\end{lemma}
\begin{proof}
Immediate from the previous lemma.
\end{proof}
\subsubsection{Converse coding appplied to subchannels}
Let $G\cong \bigoplus_{i=1}^I R_i$ be an Abelian group and let $H\cong \bigoplus_{i=1}^I p_i^{\theta_i}R_i$ be an arbitrary subgroup of $G$ and let $S$ be the optimal subchannel corresponding to the subgroup $H$. Using standard algebraic arguments we can show that for any shifted group code $\mathds{C}+v$ where $\mathds{C}\cong \bigoplus_{i=1}^I R_i^{k_i}$ and $v$ is an optimal coset selection vector, we have
\begin{align*}
&\mathds{C}_S=(\mathds{C}\cap H^n)+v=v+\bigoplus _{i=1}^I p_i^{\theta_i} R_i^{k_i}\\
&R_{\mathds{C}_S}=\frac{1}{n}\log \left|\mathds{C}\cap H^n\right|=\frac{1}{n}\sum_{i=1}^I(r_i-\theta_i)k_i\log p_i
\end{align*}
Define $w_i=\frac{r_i k_i \log p_i}{\log M}$ then we get $R_{\mathds{C}_S}=\sum_{i=1}^I \frac{r_i-\theta_i}{r_i}w_iR$.

\begin{lemma}
For a maximal subchannel $H$ of the channel $(G,\mathcal{Y},W)$, $C_{H,G}\le C_{H^*}^U$
\end{lemma}
\begin{proof}
Shannon's coverse implies
\begin{align*}
R<\frac{1}{n}\sum_{i=1}^n I(X_i,Y_i)
\end{align*}
where $X_i$'s have uniform distributions over cosets of subgroups of $H$. Since $H$ is maximal, all of these distributions result in a mutual information less than $C_{|H^*|}^U$. Therefore, the average is also less than $C_{H^*}^U$. Conclude that $R<C_{H^*}^U$.
\end{proof}
The lemma implies
\begin{align*}
R_{\mathds{C}_S}=\sum_{i=1}^I \frac{r_i-\theta_i}{r_i}w_iR<C_{|H^*|}^U
\end{align*}
Therefore, for all maximal subchannels $H$, $R<\frac{C_{|H^*|}^U}{\sum_{i=1}^I \frac{r_i-\theta_i}{r_i}w_i}$. This proves the theorem.
\section{Special cases}\label{special_cases}

\subsection{Linear Codes}
The capacity of linear codes has been studied in \cite{ahlswede_alg_codes}. We show that for the case of linear codes over $\mathds{F}_q$, the upper and lower bounds are tight and are equal to the capacity given in \cite{ahlswede_alg_codes}. Let $\mathds{C}$ be a group code over the field $\mathds{F}_q$ for some prime number $q$. Since the only subgroups of $\mathds{F}_q$ are the trivial subgroup and the group $\mathds{F}_q$ itself, the lower bound reduces to $C_{|\mathds{F}_q|}^U$; And since $\mathds{F}_q$ is maximal in itself, the upper bound also reduces to $C_{|\mathds{F}_q|}^U$. Therefore the capacity of linear codes over $\mathds{F}_q$ is given by $C_{\mathds{F}_q}=C_{|\mathds{F}_q|}^U=I(X;Y)$ where $X$ has a uniform distribution over the input alphabet.
\subsection{Symmetric Channels}
For a symmetric channel, uniform input distribution over cosets of an arbitrary subgroup $H$ of $G$ results in the same mutual information with the channel output; This means all of the cosets of $H$ are optimal and we can pick $H^*=H$. The lower bound reduces to
\begin{align*}
C_G\ge \max_{\substack{w_1,\cdots,w_I\\w_1+\cdots+w_I=1}}\min_{H\le G} \frac{C_{|H|}^U}{w_H}
\end{align*}
Since all of the subgroups are maximal for a symmetric channel, the lower bound also reduces to the same expression. i.e. The capacity of group codes over symmetric channels is given by:
\begin{align*}
C_G= \max_{\substack{w_1,\cdots,w_I\\w_1+\cdots+w_I=1}}\min_{H\le G} \frac{C_{|H|}^U}{w_H}
\end{align*}
where $w_H=\sum_{i=1}^I \frac{r_i-\theta_i}{r_i}w_i$ for $H\cong \bigoplus_{i=1}^I p_i^{\theta_i} R_i$. The capacity of Abelian group codes over symmetric channels given in \cite{fagnani_abelian} coincides with the new result.

\section{Conclusion}
In this paper, we investigated the performance limits of Abelian group codes over discrete memoryless channels. Upper and lower bounds on the capacity of such codes has been computed and we presented two special cases where the bounds match. Our results unify the known results on the capacity of structured codes for the point to point channel coding problem and states the information theoretic performance limits of structured codes based on the algebraic structure of the underlying group.

\bibliographystyle{plain}
\bibliography{ariabib}

\pagebreak
\section{Appendix}
\subsection{proof of lemma \ref{problemma}}
Let $a=u(m)-u(\tilde{m})$ and $h=x-\tilde{x}$. First assume $a\in \mathds{Z}_{p^r}^k$ and $h\in \mathds{Z}_{p^r}^n$ and let $G\in \mathds{Z}_{p^r}^{k\times n}$ be a random matrix and $v\in\mathds{Z}_{p^r}^k$ be a random vector. In order to calculate the probability $P(u(m)G+v=x,u(\tilde{m})G+v=\tilde{x})=P\left(aG=h,u(m)G+v=x\right)$ we need to count the
number of solutions of $\sum_{l=1}^{k}{a_l g_l}=h$ where $a_l$'s are the elements of $a$ and $g_l$'s are the rows of the generator matrix $G$. There are $r$ cases that we should take into account.
\begin{itemize}
\item \textbf{Case 0:} There exists an index $m$ such that $a_m \in \mathbb{Z}_{p^r}\backslash p\mathbb{Z}_{p^r}$.\\
In this case, the number of choices for $G$ and $v$ is equal to $p^{nr(k-1)}$ whereas the total number of choices is equal to $p^{nr(k+1)}$. Since the matrix G and the vector $v$ are  chosen equilikely, we have:
\begin{align*}
P\left(aG=h,uG+v=x\right)=\frac{p^{nr(k-1)}}{p^{nr(k+1)}}=\frac{1}{p^{2nr}}
\end{align*}
Let $T_0(m)$ be the set of all indices $\tilde{m}$ that fall in this
category. Then we have $|T_0(m)|=(p^r)^k-(p^{r-1})^k$.

\item \textbf{Case $\theta$ ($\theta=1,2,\ldots,r-1$):} The conditions is cases $0$ up to $\theta-1$ are not satisfied and there exists an index $m$ such that $a_m \in p^\theta\mathbb{Z}_{p^r}\backslash p^{\theta+1}\mathbb{Z}_{p^r}$.\\
In this case, if $h\notin (p^\theta \mathbb{Z}_{p^r})^n$ there are no choices fore $G$ and $v$. Otherwise, the number of choices for $G$ and $v$ is equal to $p^{n\theta} p^{nr(k-1)}$ where as the total number of choices is equal to $p^{nr(k+1)}$. Since the matrix G and the vector $v^n$ are chosen equilikely, we have:
\begin{align*}
\nonumber P(aG= & h,u^k(i)G+v^n=x^n)=\\
&\left\{ \begin{array}{ll}
\frac{p^{n\theta}p^{nr(k-1)}}{p^{nr(k+1)}}= \frac{1}{p^{n(2r-\theta)}}
&\mbox{if $\tilde{x}\in x+(p^\theta \mathbb{Z}_{p^r})^n$};\\
0 & \mbox{otherwise}.
\end{array}\right.
\end{align*}
Let $T_\theta(m)$ be the set of all indices $\tilde{m}$ that fall in this category. Then we have $|T_\theta(m)|=(p^{r-\theta})^k-(p^{r-\theta-1})^k$.
\end{itemize}
Our original problem can be addressed by using the above result for each ring $R_i$ for which we take the matrix $G_{K,N}=(g_{iK}^N)$ and replace $k$ by $w_ik$. Since the elements of $v_N$ and $g_{iK}^N$ are chosen uniformly, their components are independent across different rings $R_i$. Therefore, the joint probability is the product of probabilities for each ring and the total number of such indices is the product of the number of possible indices for each ring $R_i$. Therefore,
\begin{align*}
&\nonumber P(e(m)=x,e(\tilde{m})=\tilde{x})=\\
&\left\{ \begin{array}{ll}
\prod_{i=1}^I\frac{1}{p_i^{n(2r_i-\theta_i)}}
&\mbox{if $\tilde{x}\in x+\left[\bigoplus_{i=1}^I p_i^{\theta_i} R_i\right]^n$};\\
0 & \mbox{otherwise}.
\end{array}\right.
\end{align*}
Alternatively,
\begin{align*}
&\nonumber P(e(\tilde{m})=\tilde{x}|e(m)=x)=\\
&\left\{ \begin{array}{ll}
\prod_{i=1}^I\frac{1}{p_i^{n(r_i-\theta_i)}}
&\mbox{if $\tilde{x}\in x+\left[\bigoplus_{i=1}^I p_i^{\theta_i} R_i\right]^n$};\\
0 & \mbox{otherwise}.
\end{array}\right.
\end{align*}
Moreover, for a fixed $m$, let $T_\theta(m)$ be the set of all $\tilde{m}$ with $\theta(m,\tilde{m})=(\theta_1,\theta_2,\cdots,\theta_I)$, then
\begin{align*}
|T_\theta(m)|=\prod_{i=1}^I\left[(p_i^{r_i-\theta_i})^{w_ik}(1-p_i^{-w_ik})\right]
\end{align*}
Therefore,
\begin{align*}
|T_\theta(m)|\le\prod_{i=1}^I\left[(p_i^{r_i-\theta_i})^{w_ik}\right]
\end{align*}
This result can also be confirmed by the straightforward method.

\subsection{proof of \ref{cosettypicallemma}}
Let $x,y,z\in G^n$ where $x=(x_1,\cdots,x_I)$, $y=(y_1,\cdots,y_I)$ and $z=(z_1,\cdots,z_I)$ where $x_i,y_i,z_i\in R_i^n$. For $i=1,\cdots,I$, define
\begin{align*}
&S=\left(x+\left[\bigoplus_{i=1}^I p_i^{\theta_i} R_i\right]^n\right)\cap A_\epsilon^n(X|y)\\
&S_i=\left(x_i+\left[p_i^{\theta_i} R_i\right]^n\right)\cap A_\epsilon^n(X_i|y)\\
\end{align*}
where $X$ is uniform over $G$ and $X_i$ is uniform over $R_i$. First we show that $S\subseteq S_1\times\cdots S_I$. Let $z\in S$; Since $z\in A_\epsilon^n(X|y)$,
\begin{align*}
\left|\frac{1}{n}N(a,b|z,y)-p_{XY}(a,b)\right|\le \frac{\epsilon}{|G|\cdot |\mathcal{Y}|}
\end{align*}
for arbitrary $a=(a_1,\cdots,a_I)\in G$ and $b=(b_1,\cdots,b_I)\in G$. We have:
\begin{align*}
&\left|\frac{1}{n}N(a_1,b|z_1,y)-p_{X_1Y}(a_1,b)\right|\\
&=\left|\frac{1}{n}\sum_{a_2,\cdots,a_I}N(a,b|z,y)-\sum_{a_2,\cdots,a_I}p_{XY}(a,b)\right|\\
&\le \sum_{a_2,\cdots,a_I}\left|\frac{1}{n}N(a,b|z,y)-p_{XY}(a,b)\right|\\
&\le |R_2|\times\cdots\times |R_I|\times\frac{\epsilon}{|G|\cdot |\mathcal{Y}|}\\
&=\frac{\epsilon}{|R_1|\cdot |\mathcal{Y}|}
\end{align*}
We conclude that $z_1\in A_\epsilon^n(X_i|Y)$. Since $z\in \left[\bigoplus_{i=1}^I p_i^{\theta_i} R_i\right]^n$, we have $z_1\in \left[p_1^{\theta_1}R_1\right]^n$. Therefore, $z_1\in S_1$. This is true for all $i=1,\cdots,I$. Conclude that $S\subseteq S_1\times\cdots\times S_I$. Therefore,
\begin{align*}
|S|\le |S_1|\times\cdots\times |S_I|
\end{align*}

By Lemma 5 of \cite{dinesh_dsc}, we have:
\begin{align*}
|S_i|\le 2^{n\left[H(X_i|Y)-H([X_i]_{\theta_i})+O(\epsilon)\right]}
\end{align*}

Conclude that
\begin{align*}
\left|\left(x+\left[\bigoplus_{i=1}^I p_i^{\theta_i} R_i\right]^n\right)\cap A_\epsilon^n(y)\right|\le\prod_{i=1}^I2^{n\left[H(X_i|Y)-H([X_i]_{\theta_i})+O(\epsilon))\right]}
\end{align*}

\subsection{proof of lemma \ref{upperbound}}
The expected value of the average probability of word error is given by:
\begin{align*}
&\nonumber\mathds{E}\left\{P_{avg}(err)\right\}=\sum_{m=1}^{M}\frac{1}{M}\sum_{x\in
G^n}P\left(e(m)=x\right) \sum_{\substack{\tilde{m}=1\\\tilde{m}\ne m}}^{M}\\
&\sum_{y\in A_\epsilon^n(Y|x)} \sum_{\tilde{x}\in A_\epsilon^n(X|y)}P\left(e(\tilde{m})=\tilde{x}, Y^n=y|e(m)=x\right)+O(\epsilon)\\
&\nonumber =\sum_{m=1}^{M}\frac{1}{M}\sum_{x\in G^n}
P(e(m)=x)\sum_{\substack{\tilde{m}=1\\\tilde{m}\ne m}}^{M}\sum_{y\in A_\epsilon^n(Y|x)} \sum_{\tilde{x}\in A_\epsilon^n(X|y)}\\
&P(e(\tilde{m})=\tilde{x}|e(m)=x)W^n(y|x)+O(\epsilon)
\end{align*}
Using Lemma \ref{problemma} and Lemma \ref{cosettypicallemma} we get:
\begin{align*}
&\nonumber\mathds{E}\left\{P_{avg}(err)\right\}\le\\
&\nonumber \sum_{m=1}^{M}\frac{1}{M}\sum_{x\in G^n}
P(e(m)=x)\sum_{\theta}\sum_{\tilde{m}\in T_\theta(m)}\sum_{y\in A_\epsilon^n(Y|x)}\\
&\sum_{\tilde{x}\in \left(x+\left[\bigoplus_{i=1}^I p_i^{\theta_i} R_i\right]^n\right)\cap A_\epsilon^n(y)}\prod_{i=1}^I\frac{1}{p_i^{n(r_i-\theta_i)}}W^n(y|x)+O(\epsilon)\\
&\le \sum_{m=1}^{M}\frac{1}{M}\sum_{x\in G^n}
P(e(m)=x)\sum_{\theta}\sum_{\tilde{m}\in T_\theta(m)}\sum_{y\in A_\epsilon^n(Y|x)}\\
&\prod_{i=1}^I2^{n\left[H(X_i|Y)-H([X_i]_{\theta_i})+O(\epsilon)\right]}\prod_{i=1}^I\frac{1}{p_i^{n(r_i-\theta_i)}}W^n(y|x)+O(\epsilon)
\end{align*}

\begin{align*}
&= \sum_{m=1}^{M}\frac{1}{M}\sum_{x\in G^n}
P(e(m)=x)\sum_{\theta}\sum_{\tilde{m}\in T_\theta(m)}\\
&\prod_{i=1}^I 2^{n\left[H(X_i|Y)-H([X_i]_{\theta_i})+O(\epsilon)\right]}\prod_{i=1}^I\frac{1}{p_i^{n(r_i-\theta_i)}}\sum_{y\in A_\epsilon^n(Y|x)}W^n(y|x)+O(\epsilon)\\
& \le \sum_{m=1}^{M}\frac{1}{M}\sum_{x\in G^n}
P(e(m)=x)\sum_{\theta}\sum_{\tilde{m}\in T_\theta(m)}\\
&\prod_{i=1}^I2^{n\left[H(X_i|Y)-H([X_i]_{\theta_i})+O(\epsilon)\right]}\frac{1}{p_i^{n(r_i-\theta_i)}}+O(\epsilon)\\
&= \sum_{\theta}\sum_{m=1}^{M}\frac{1}{M}\sum_{x\in G^n}P(e(m)=x)\prod_{i=1}^I\left[(p_i^{r_i-\theta_i})^{w_ik}\right]\\
&\prod_{i=1}^I \left[2^{n\left[H(X_i|Y)-H([X_i]_{\theta_i})+O(\epsilon)\right]}\frac{1}{p_i^{n(r_i-\theta_i)}}\right]+O(\epsilon)\\
&\le \sum_{\theta}\prod_{i=1}^I\left[(p_i^{r_i-\theta_i})^{w_ik}2^ {n\left[H(X_i|Y)-H([X_i]_{\theta_i})+O(\epsilon)\right]}\frac{1}{p_i^{n(r_i-\theta_i)}}\right]+O(\epsilon)\\
\end{align*}
Since $\epsilon$ is arbitrary let $\epsilon\rightarrow 0$ and therefore $O(\epsilon)\rightarrow 0$ to get:
\begin{align*}
\mathds{E}\left\{P_{avg}(err)\right\}\le \sum_\theta \exp_2\left\{-n\sum_{i=1}^I \left[(r_i-\theta_i)\log p_i\right.\right.&\\
\left.\left.-\frac{w_i k}{n}(r_i-\theta_i)\log p_i-H(X_i|Y)+H([X_i]_{\theta_i}|Y)\right]\right\}&\\
\end{align*}
\end{document}